\documentclass[11pt]{article}
\usepackage[utf8]{inputenc}
\usepackage[T1]{fontenc}
\usepackage{authblk}
\usepackage{amsmath,amssymb,amsfonts,amsthm}
\usepackage[top=1in,bottom=1in,left=0.75in,right=0.75in]{geometry}
\usepackage{xspace,graphicx}

\newtheorem{proposition}{Proposition}

\theoremstyle{definition}
\newtheorem{definition}{Definition}[section]

\theoremstyle{remark}

\newcommand{\ket}[1]{\ensuremath{|#1\rangle}\xspace}
\newcommand{\bra}[1]{\ensuremath{\langle #1|}\xspace}
\newcommand{\Braket}[3]{\ensuremath{\bra{#1}#2\ket{#3}}\xspace}
\newcommand{\braket}[2]{\ensuremath{\langle #1|#2\rangle}\xspace}
\numberwithin{equation}{section}
\title{Generalized squeezed-coherent states of the finite one-dimensional oscillator and matrix multi-orthogonality}
\date{}
\author[1]{Vincent X. Genest\thanks{genestvi@crm.umontreal.ca}}
\author[1]{Luc Vinet\thanks{luc.vinet@umontreal.ca}}
\author[2]{Alexei Zhedanov\thanks{zhedanov@kinetic.ac.donetsk.ua}}
\affil[1]{Centre de recherches math\'ematiques, Universit\'e de Montr\'eal, C.P. 6128, Succursale Centre-ville, Montr\'eal, Qu\'ebec, H3C 3J7, Canada}
\affil[2]{Donetsk Institute for Physics and Technology, Donetsk 83114, Ukraine}
\begin{document}
\maketitle
\thispagestyle{empty}
\hrule
\begin{abstract}
A set of generalized squeezed-coherent states for the finite $\mathfrak{u}(2)$ oscillator is obtained. These states are given as linear combinations of the mode eigenstates with amplitudes determined by matrix elements of exponentials in the $\mathfrak{su}(2)$ generators. These matrix elements are given in the $(N+1)$-dimensional basis of the finite oscillator eigenstates and are seen to involve $3\times 3$ matrix multi-orthogonal polynomials $Q_n(k)$ in a discrete variable $k$ which have the Krawtchouk and vector-orthogonal polynomials as their building blocks. The algebraic setting allows for the characterization of these polynomials and the computation of mean values in the squeezed-coherent states. In the limit where $N$ goes to infinity and the discrete oscillator approaches the standard harmonic oscillator, the polynomials tend to $2\times 2$ matrix orthogonal polynomials and the squeezed-coherent states tend to those of the standard oscillator.
\end{abstract}
\textbf{Keywords:} Finite quantum oscillator, matrix orthogonality, vector polynomials, Krawtchouk polynomials, $d$-orthogonal polynomials, squeezed and coherent states, $\mathfrak{u}(2)$ algebra.
\bigskip

\hrule
\section{Introduction}
Discretizations of the standard quantum harmonic oscillator are provided by finite oscillator models (see for instance \cite{Wolf-2001,Spindel-2008}). We here consider the one based on the Lie algebra $\mathfrak{u}(2)=\mathfrak{u}(1)\oplus \mathfrak{su}(2)$ which has been interpreted as a quantum optical system consisting of $N+1$ equally spaced sensor points \cite{Wolf-2001}. In this connection, we investigate here the matrix elements of exponentials of linear and quadratic expressions in the $\mathfrak{su}(2)$ generators; these operators represent discrete analogues of the squeeze-coherent states operators for the standard quantum oscillator. As shall be seen, these matrix elements are given in terms of matrix multi-orthogonal polynomials.

These polynomials (defined below), generalize the standard orthogonal polynomials by being orthogonal with respect to a \emph{matrix} of functionals \cite{Sorokin-1997}. Very few explicit examples have been encountered in the literature; remarkably, our study entails a family of such polynomials and the algebraic setting allows for their characterization.
\subsection{Finite oscillator and $\mathfrak{u}(2)$ algebra}
The standard one-dimensional quantum oscillator is described by the Heisenberg algebra $\mathfrak{h}_1$, with generators $a$, $a^{\dagger}$ and $\mathfrak{id}$ obeying 
\begin{equation}
[a,a^{\dagger}]=\mathfrak{id}\;\;\text{and}\;\;[a,\mathfrak{id}]=[a^{\dagger},\mathfrak{id}]=0.
\end{equation}
The Hamiltonian is given by $H=a^{\dagger}a+1/2$ and with the position operator $Q$ and the momentum operator $P$ defined as follows:
\begin{align}
Q=\frac{1}{2}(a+a^{\dagger}),\;\;\;\; P=-\frac{i}{2}(a-a^{\dagger}),
\end{align}
the equations of motion
\begin{align}
\label{eq-motion-1}
[H,Q]&=-iP,\\
\label{eq-motion-2}
[H,P]&=iQ,
\end{align}
are recovered.

The finite oscillator model is obtained by replacing the Heisenberg algebra by the algebra $\mathfrak{u}(2)=\mathfrak{u}(1)\oplus\mathfrak{su}(2)$. The $\mathfrak{su}(2)$ generators are denoted by $J_{1}$, $J_{2}$ and $J_{3}$ and verify
\begin{align}
[J_{i},J_{j}]=i\epsilon_{ijk}J_{k},
\end{align}
with $\epsilon_{ijk}$ the Levi-Civita symbol. The $\mathfrak{u}(1)$ generator is later to be $\frac{N}{2}\mathfrak{id}$. For the finite oscillator, the correspondence with the physical "observables" is as follows:
\begin{align}
&\text{Position operator:}\;\; Q=J_{1}, \\
&\text{Momentum operator:}\;\; P=-J_{2}, \\
&\text{Hamiltonian:}\;\; H=J_{3}+\frac{(N+1)}{2}\,\mathfrak{id}.
\end{align}
While this relaxes the functional dependence of the Hamiltonian, it is readily seen that this identification reproduces the Hamilton-Lie equations \eqref{eq-motion-1} and \eqref{eq-motion-2}.

In quantum optics, such a system can be identified with signals coming from an array of $N+1$ sensor points \cite{Wolf-2001}. The states of this system can be expanded in the eigenbasis of the Hamiltonian $H=J_{3}+N/2+1/2$, which spans the vector space of the $(N+1)$-dimensional unitary irreducible representation of the $\mathfrak{su}(2)$ algebra. The eigenstates of $H$ are denoted $\ket{N,n}$ and one has
\begin{align}
H\ket{N,n}=(n+1/2)\ket{N,n},
\end{align}
with $n=0,\ldots,N$. The number $n$ will often be referred to as the mode number and the states $\ket{N,n}$ as the mode eigenstates. This oscillator model thus only has a finite number of excitations, as opposed to an infinite number for the standard oscillator. Moreover, in this representation, the spectrum of the momentum and position operators $P$ and $Q$ consists of equally-spaced discrete values ranging from $-N/2$ to $N/2$. The position and momentum eigenbases can be obtained from the mode eigenbasis by simple rotations and their overlaps are $\mathfrak{su}(2)$ Wigner functions \cite{Wolf-2001}.

 It is convenient to introduce the usual shift operators $J_{\pm}$ and the number operator $\widehat{N}$. These operators are defined by
\begin{align}
J_{\pm}&=(J_1\pm \mathrm{i}J_2),\\
\widehat{N}&=J_3+N/2.
\end{align}
The action of these operators on the mode eigenstates is given by
\begin{align}
J_{+}\ket{N,n}&=\sqrt{(n+1)(N-n)}\ket{N,n+1},\\
J_{-}\ket{N,n}&=\sqrt{n(N-n+1)}\ket{N,n-1},\\
\widehat{N}\ket{N,n}&=n\ket{N,n}.
\end{align}
For the shift operators $J_{\pm}$, the action of any of their positive powers has the form
\begin{align}
\label{action-1}
J_{+}^{\alpha}\,\ket{N,n}&=\sqrt{\frac{(n+\alpha)!(N-n)!}{n!(N-n-\alpha)!}}\,\ket{N,n+\alpha}=\sqrt{(-1)^{\alpha}(n+1)_{\alpha}(n-N)_{\alpha}}\,\ket{N,n+\alpha},\\
\label{action-2}
J_{-}^{\beta}\,\ket{N,n}&=\sqrt{\frac{n!(N-n+\beta)!}{(n-\beta)!(N-n)!}}\ket{N,n-\beta}=\sqrt{(-1)^{\beta}(-n)_{\beta}(N-n+1)_{\beta}}\,\ket{N,n-\beta},
\end{align}
where $(n)_{0}=1$ and $(n)_{\alpha}=n(n+1)\cdots(n-\alpha+1)$ stands for the Pochhammer symbol. It is worth noting that in contradistinction with the standard quantum harmonic oscillator, the finite oscillator possesses both a ground state and an anti-ground state. Indeed, one has $J_{+}\ket{N,N}=0$ and $J_{-}\ket{N,0}=0$. This symmetry will play a role in what follows.
\subsection{Contraction to the standard oscillator}
In the limit $N\rightarrow \infty$, the finite oscillator tends to the standard quantum harmonic oscillator through the contraction of $\mathfrak{u}(2)$ to $\mathfrak{h}_{1}$ \cite{Wolf-2003,Wigner-1956}. In this limit, after an appropriate rescaling, the shift operators $J_{\pm}$ tend to the operators $a^{\dagger}$ and $a$. Precisely, with
\begin{align}
\lim_{N\rightarrow \infty}\frac{J_{+}}{\sqrt{N}}=a^{\dagger}, &&\lim_{N\rightarrow \infty}\frac{J_{-}}{\sqrt{N}}=a,
\end{align}
the commutation relation $[a,a^{\dagger}]=\mathfrak{id}$ of the Heisenberg-Weyl algebra $\mathfrak{h}_1$ is recovered. Moreover, the contraction of the Hamiltonian $H$ leads to the standard quantum oscillator Hamiltonian $H_{\text{osc}}=\frac{1}{2}(P^2+Q^2)$. This limit shall be used to establish the correspondence with studies associated with the standard harmonic oscillator \cite{Vinet-2011}.
\subsection{Exponential operator and generalized coherent states}
It is known that the standard one-dimensional harmonic oscillator admits the Schr\"odinger algebra $\mathfrak{sh}_1$ as dynamical algebra \cite{Niederer-1972,Niederer-1973}. This algebra is generated by the linears and bilinears in $a$ and $a^{\dagger}$, that is $a$, $a^{\dagger}$, $\mathfrak{id}$, $a^{2}$, $(a^{\dagger})^2$ and $a^{\dagger}a$. The representation of the group $Sch_{1}$ has been recently constructed and analyzed in the oscillator state basis by two of us \cite{Vinet-2011}. It involved determining the matrix elements of the exponentials of linear and quadratic expressions in $a$ and $a^{\dagger}$. The study hence had a direct relation to the generalized squeezed-coherent states of the ordinary quantum oscillator \cite{Satyanarayana-1985} .

We pursue here a similar analysis for the finite oscillator. Notwithstanding the fact that the linears and bilinears in $J_{+}$ and $J_{-}$ no longer form a Lie algebra, our purpose is to determine analogously the matrix elements of the fully disentangled exponential operator
\begin{align}
\label{Operator-R}
R(\eta,\xi)=D(\eta)\cdot S(\xi)=e^{\eta J_{+}}e^{\mu J_{3}}e^{-\overline{\eta}J_{-}}\cdot e^{\xi J_{+}^{2}/2}e^{-\overline{\xi} J_{-}^{2}/2},
\end{align}
 in the basis of the finite oscillator's states. The parameters $\eta$ and $\xi$ are complex-valued and $\mu=\log(1+\eta\overline{\eta})$. The matrix elements in the $(N+1)$-dimensional eigenmode basis shall be denoted 
\begin{align}
\label{Operator-R-Elements}
R_{k,n}=\Braket{k,N}{R(\eta,\xi)}{N,n}.
\end{align}
In parallel with the definition of the standard harmonic oscillator squeezed-coherent states, we introduce the following normalized set of states
\begin{equation}
\label{Coherent-States-1}
\ket{\eta,\xi}:=\frac{1}{|\braket{\eta,\xi}{\eta,\xi}|^2}R(\eta,\xi)\ket{N,0},
\end{equation}
which are a special case of the generalized coherent states
\begin{equation}
\ket{\eta,\xi}_{n}:=\frac{1}{A}\sum_{k}R_{k,n}\ket{N,k},
\end{equation}
where $A$ is a normalization factor. 

Contrary to the case of the harmonic oscillator, the operator $R$ considered here is not unitary. While the operator $D(\eta)$ can be shown to be unitary \cite{Truax-1985}, such is not the case for the operator $S(\xi)$. Nonetheless, we shall observe that the superpositions of states in \eqref{Coherent-States-1} show spin squeezing and entanglement according to the criteria found \cite{Wang-2011} and \cite{Wang-2001}. In addition, the consideration of the fully disentangled form \eqref{Operator-R} allows for the explicit calculation of the matrix elements in terms of known polynomials, which is not possible with other choices of the squeezing operator for the finite oscillator \cite{Wolf-2007}.

As previously mentioned, the matrix elements \eqref{Operator-R-Elements} will be naturally expressed in terms of a finite family of $3\times 3$ multi-orthogonal matrix polynomials $Q_{n}(k)$ in the discrete variable $k$. In the contraction limit, these polynomials tend to the $2\times 2$ matrix orthogonal polynomials encountered in \cite{Vinet-2011}.
\subsection{Matrix multi-orthogonality}
Matrix multi-orthogonality has been first studied in the context of Pad\'e-type approximation \cite{Beckermann-1992}. The algebraic aspects of matrix multi-orthogonality (recurrence relation, Shohat-Favard theorem, Darboux transformation, etc.) are discussed by Sorokin and Van Iseghem in \cite{Sorokin-1997}. Their study is based on matrix orthogonality for vector polynomials. We shall here recall the basic results to be used in what follows.

We first introduce the canonical basis for the vector space of vector polynomials of size $q$:
\begin{equation}
e_{0}=
\begin{pmatrix}
1 \\
0 \\
\vdots \\
0    
\end{pmatrix},\; \ldots,\;
e_{q-1}=\begin{pmatrix}
0 \\
0 \\
\vdots \\
1    
\end{pmatrix},\;
e_{q}=\begin{pmatrix}
x \\
0 \\
\vdots \\
0    
\end{pmatrix},\;\ldots,\;
e_{2q-1}=\begin{pmatrix}
0 \\
0 \\
\vdots \\
x    
\end{pmatrix},
e_{2q}=\begin{pmatrix}
x^2 \\
0 \\
\vdots \\
0    
\end{pmatrix},\;\ldots
\end{equation}
For $i=\lambda q+s$ with $i\geqslant 0$ and $s=0,\ldots,q-1$, the basis vector $e_{i}$ has the component $x^{\lambda}$ in the $s+1^{\text{th}}$ position and zeros everywhere else. A vector polynomial of the form $\alpha_{0}e_{0}+\cdots+\alpha_{n}e_{n}$ with $\alpha_{n}\neq 0$ will be said of order $n$. If $q=1$, this corresponds to a standard polynomial of degree $n$ in $x$. Let $q$ and $p$ be positive integers and $H_{n}(x)=(h_{n,1}(x),h_{n,2}(x),\ldots,h_{n,q}(x))^{t}$ be a $q$-vector polynomial of order $n$, where $t$ denotes the transpose \cite{Sorokin-1997}. The vector polynomial $H_{n}(x)$ is multi-orthogonal if there exists a $p\times q$ matrix of functionals $\Theta_{i,j}$ with $i=0,\ldots,p$ and $j=0,\ldots,q$, defined by their moments, such that the following relations hold:
\begin{align}
\label{orthogonality-condition-1}
\Theta_{1,1}(h_{n,1}(x)x^{\nu})+\cdots+\Theta_{1,q}(h_{n,q}(x)x^{\nu})&=0, && \nu=0,\ldots,n_1-1,\\
\cdots\nonumber\\
\label{orthogonality-condition-2}
\Theta_{p,1}(h_{n,1}(x)x^{\nu})+\cdots+\Theta_{p,q}(h_{n,q}(x)x^{\nu})&=0, && \nu=0,\ldots,n_p-1,
\end{align}
The numbers $(n_1,\ldots,n_{p})$ are defined as follows: set $n=\mu\,p+\delta$, with $\delta=0,\ldots,p-1$, then $n_1=\cdots=n_{\delta}=\mu+1$ and $n_{\delta+1}=\cdots=n_{p}=\mu$. 

It was shown \cite{Sorokin-1997} that such polynomials obey a recurrence relation of the form
\begin{align}
\label{recurrence-generale-1}
c_{n}^{(q)}H_{n+q}(x)&+\cdots+c_{n}^{(1)}H_{n+1}(x)+c_{n}^{(0)}H_{n}(x)\nonumber\\
&+c_{n}^{(-1)}H_{n-1}(x)+\cdots+c_{n}^{(-p)}H_{n-p}(x)=xH_{n}(x),
\end{align}
along with the initial conditions $H_{-p}=\cdots=H_{-1}=0$; it was also proven \cite{Sorokin-1997} that a recurrence of the type \eqref{recurrence-generale-1} implies the orthogonality conditions \eqref{orthogonality-condition-1} and \eqref{orthogonality-condition-2}.

These relations are more easily handled by introducing matrix polynomials, which are matrices of polynomials. Suppose that $p\geqslant q$, the matrix polynomials are obtained by first writing the recurrence relation \eqref{recurrence-generale-1} for $k$ consecutive indices. One has
\begin{align}
x(H_{n}(x),\ldots,H_{n+k-1}(x))=(H_{n-p}(x),\ldots,H_{n+k-1+q}(x))
\begin{pmatrix}
c_{n}^{(-p)} & & \\
\vdots &\ddots & c_{n+k-1}^{(-p)}\\
c_{n}^{(0)} & \cdots & c_{n+k-1}^{(-q)}\\
\vdots & \ddots & c_{n+k-1}^{(0)}\\
c_{n}^{(q)}& \ddots & \vdots\\
&&c_{n+k-1}^{(q)}
\end{pmatrix}.
\end{align}
One can choose $k$ to be the greatest common divisor of $p$ and $q$; in this case, we can set $p=\sigma k$, $q=\rho k$ and the matrix on the right hand side can be put in blocks of size $k\times k$. We define the $q\times k$ matrix polynomial by $Q_{n}(x)=(H_{n k}(x),\ldots,H_{nk+k-1}(x))$. The recurrence relation \eqref{recurrence-generale-1} thus becomes
\begin{equation}
\label{recurrence-generale-2}
x\,Q_{n}(x)=\sum_{\ell=-\sigma}^{\rho}\Gamma_{n}^{(\ell)}Q_{n+\ell}(x).
\end{equation}
At the end points $-\sigma$ and $\rho$ in the sum, the matrix coefficient $\Gamma_{n}^{-\sigma}$ is an upper triangular invertible matrix and $\Gamma_{n}^{\rho}$ is a lower triangular invertible matrix. For $q=1$, this recurrence relation characterizes vector orthogonality of order $p$, also called $p$-orthogonality. If $p=q$ and $\Gamma_{n}^{-\sigma}=(\Gamma_{n}^{\rho})^{*}$, matrix orthogonality is recovered. In this paper, the special case corresponding to $q=3$ and $p=9$ will be encountered.
\subsection{Outline}
The outline of the paper is as follows. In section 2, we obtain the recurrence relation satisfied by the matrix elements $R_{k,n}$ and show that they involve $3\times 3$ matrix multi-orthogonal polynomials. In section 3, we express these matrix elements as a finite convolution involving the Krawtchouk polynomials and a family of $3$-orthogonal polynomials recently studied in \cite{VXGenest-2011}. In section 4, we obtain a biorthogonality relation for the matrix polynomials $Q_{n}(x)$. In section 5, we calculate the matrix orthogonality functionals $\Theta_{i,j}$ for the polynomials. In section 6, we obtain a difference equation for the polynomials and discuss the dual picture. In section 7, we derive the generating functions and ladder relations. In section 8, we discuss the properties of the states $\ket{\eta,\xi}$ and study spin squeezing in this system. In section 9, we briefly review the contraction limit $N\rightarrow \infty$ and relate our results with those of \cite{Vinet-2011}. We close with concluding remarks in section 10. Appendices containing $\mathfrak{su}(2)$ structure formulas and properties of the Krawtchouk polynomials are included.
\newpage
\section{Recurrence relation}
We shall begin the analysis by obtaining the recurrence relation satisfied by the matrix elements $R_{k,n}$. We first observe that
\begin{equation}
\label{recurrence-step-1}
(k-N/2)R_{k,n}=\Braket{k,N}{J_{3}R}{N,n}=\Braket{k,N}{RR^{-1}J_{3}R}{N,n},
\end{equation}
where the inverse operator is given by
\begin{equation}
R^{-1}=e^{\frac{\overline{\xi}}{2}J_{-}^2}e^{-\frac{\xi}{2}J_{+}^2}e^{\overline{\eta}J_{-}}e^{-\mu J_{3}}e^{-\eta J_{+}}.
\end{equation}
The recurrence relation is obtained by expanding the expression $R^{-1}J_{3}R$ and acting on the eigenstates $\ket{N,n}$. This is done using the Baker--Campbell--Hausdorff relation
\begin{equation}
e^{A}Be^{-A}=B+[A,B]+\frac{1}{2!}[A,[A,B]]+\frac{1}{3!}[A,[A,[A,B]]]+\ldots,
\end{equation}
and the formulas for $\mathfrak{su}(2)$ found in appendix A. Using the polar parametrizations $\eta=\rho e^{i\delta}$ and $\xi=r e^{i\gamma}$, we obtain
\begin{align}
R^{-1}J_{3}R&=(1-2p)\mathcal{A}_{0}+\rho e^{-i\delta}(1-p)\mathcal{A}_{-}+\rho (1-p)[e^{i\delta}-re^{-i(\delta-\gamma)}]\mathcal{A}_{+}\nonumber\\
&+(1-2p)r e^{i\gamma}\mathcal{A}_{+}^{2}-\rho\, r^2 e^{-i(\delta-2\gamma)}(1-p)\mathcal{A}_{+}^3-2\,\rho\,r e^{-i(\delta-\gamma)}(1-p)\mathcal{A}_{+}\mathcal{A}_{0}
\end{align}
where we have defined
\begin{align}
p&=\frac{\rho^2}{1+\rho^2}, &&\mathcal{A}_{0}=J_{3}+\overline{\xi}\,J_{-}^2, \\
\mathcal{A}_{-}&=J_{-}, &&\mathcal{A}_{+}=J_{+}-\overline{\xi}(1+2J_{3})J_{-}-\overline{\xi}^2J_{-}^3.
\end{align}
Introducing this result into \eqref{recurrence-step-1}, we find that the matrix elements $R_{k,n}$ obey the recurrence relation
\begin{align*}
\label{recurrence-step-2}
c_{n}^{(3)}R_{k,n+3}&+c_{n}^{(2)}R_{k,n+2}+c_{n}^{(1)}R_{k,n+1}+c_{n}^{(0)}R_{n,k}\\
&+c_{n}^{(-1)}R_{k,n-1}+\cdots+c_{n}^{(-9)}R_{k,n-9}=k\;R_{k,n}.
\end{align*}
The coefficients $c_{n}^{(j)}$ can be obtained straightforwardly with the help of a symbolic computation software, their explicit expressions are cumbersome and will thus be omitted here. This recurrence relation is of the form \eqref{recurrence-generale-1} with $q=3$ and $p=9$; consequently, we look for an expression of the matrix elements $R_{k,n}$ as vector polynomials. From the shape of the recurrence relation, it is natural to define the 3-vector matrix elements
\begin{align}
\Psi_{k,n}=(R_{k,3n},R_{k,3n+1},R_{k,3n+2})^{t},
\end{align}
generated by the $3\times 3$ matrix polynomial $Q_{n}(k)$,
\begin{align}
\Psi_{k,n}=Q_{n}(k)\Psi_{k,0}.
\end{align}
With these definitions, the recurrence relation \eqref{recurrence-step-2} for the matrix elements $R_{k,n}$ can be expressed as a recurrence relation for the matrix polynomial $Q_{n}(k)$. We have
\begin{equation}
k\,Q_{n}(k)=\sum_{j=-3}^{1}\Gamma_{n}^{(j)}\;Q_{n+j}(k),
\end{equation}
where the matrices $\Gamma_{n}^{(j)}$ are expressed in terms of the coefficients $c_{n}^{(j)}$ in the following manner
\begin{align}
\Gamma_{n}^{(1)}&=\begin{pmatrix} 
c_{3n}^{(3)} & 0 & 0 \\[0.1cm]
c_{3n+1}^{(2)} & c_{3n+1}^{(3)} & 0\\[0.1cm]
c_{3n+2}^{(1)} & c_{3n+2}^{(2)}  & c_{3n+2}^{(3)} 					      
\end{pmatrix}, &&\Gamma_{n}^{(0)}=\begin{pmatrix} 
c_{3n}^{(0)} & c_{3n}^{(1)} & c_{3n}^{(2)} \\[0.1cm]
c_{3n+1}^{(-1)} & c_{3n+1}^{(0)} & c_{3n+1}^{(1)}\\[0.1cm]
c_{3n+2}^{(-2)} & c_{3n+2}^{(-1)}  & c_{3n+2}^{(0)} 					      
\end{pmatrix},\\
\Gamma_{n}^{(-1)}&=\begin{pmatrix} 
c_{3n}^{(-3)} & c_{3n}^{(-2)}  & c_{3n}^{(-1)}  \\[0.1cm]
c_{3n+1}^{(-4)} & c_{3n+1}^{(-3)} & c_{3n+1}^{(-2)} \\[0.1cm]
c_{3n+2}^{(-5)} & c_{3n+2}^{(-4)}  & c_{3n+2}^{(-3)} 					      
\end{pmatrix}, &&
\Gamma_{n}^{(-2)}=\begin{pmatrix} 
c_{3n}^{(-6)} & c_{3n}^{(-5)} & c_{3n}^{(-4)} \\[0.1cm]
c_{3n+1}^{(-7)} & c_{3n+1}^{(-6)} & c_{3n+1}^{(-5)}\\[0.1cm]
c_{3n+2}^{(-8)} & c_{3n+2}^{(-7)}  & c_{3n+2}^{(-6)} 					      
\end{pmatrix},\end{align}
\begin{equation}
\Gamma_{n}^{(-3)}=\begin{pmatrix} 
c_{3n}^{(-9)} & c_{3n}^{(-8)} & c_{3n}^{(-7)} \\[0.1cm]
0& c_{3n+1}^{(-9)} & c_{3n+1}^{(-8)}\\[0.1cm]
0& 0  & c_{3n+2}^{(-9)} 					      
\end{pmatrix}.
\end{equation}
The structure of the matrix polynomial $Q_{n}(k)$ implies the existence of orthogonality relations of the form \eqref{orthogonality-condition-1} and \eqref{orthogonality-condition-2}. Equivalently, we will explicitly obtain a set of three $3\times 3$ matrices of functionals, denoted by $\mathcal{F}_{i}$, with orthogonality conditions
\begin{align}
\mathcal{F}_{i}[Q_{n}(k)k^{\nu}]=0\;\;\text{for }\;\nu=0,\ldots,\lfloor\frac{n-i}{3}\rfloor,
\end{align}
where $\lfloor x\rfloor$ denotes the integer part of $x$ and where the index $i$ runs from $1$ to $3$. These orthogonality functionals will be computed from the matrix elements of the inverse operator $R^{-1}$ and will make explicit the multi-orthogonal nature of the matrix polynomials $Q_{n}(k)$.
\section{Decomposition of matrix elements}
The matrix elements $R_{k,n}$ can be expressed as a finite convolution involving the classical Krawt\-chouk polynomials and a family of vector orthogonal polynomials studied  recently in \cite{VXGenest-2011}. Indeed, one writes 
\begin{align}
R_{k,n}&=\Braket{k,N}{R}{N,n},\nonumber\\
&=\Braket{k,N}{e^{\eta J_{+}}e^{\mu J_{3}}e^{-\overline{\eta}J_{-}}\cdot e^{\xi J_{+}^{2}/2}e^{-\overline{\xi} J_{-}^{2}/2}}{N,n},\nonumber\\
&=\sum_{m=0}^{N}\Braket{k,N}{e^{\eta J_{+}}e^{\mu J_{3}}e^{-\overline{\eta}J_{-}}}{N,m}\Braket{m,N}{e^{\xi J_{+}^{2}/2}e^{-\overline{\xi} J_{-}^{2}/2}}{N,n},\nonumber\\
&=\sum_{m=0}^{N}\lambda_{k,m}\phi_{m,n},
\end{align}
where we have defined the auxiliary matrix elements 
\begin{align}
\lambda_{k,m}&=\Braket{k,N}{e^{\eta J_{+}}e^{\mu J_{3}}e^{-\overline{\eta}J_{-}}}{N,m}, \\
\phi_{m,n}&=\Braket{m,N}{e^{\xi J_{+}^{2}/2}e^{-\overline{\xi} J_{-}^{2}/2}}{N,n}.
\end{align}
The properties of these intermediary object shall prove useful to further characterize the matrix polynomials $Q_{n}(k)$ and will also yield the explicit expansion of the states $\ket{\eta,\xi}$ in the mode eigenbasis.
\subsection{The matrix elements $\lambda_{k,m}$ and Krawtchouk polynomials}
We first study the matrix elements $\lambda_{k,m}$ of the coherent state operator $D(\eta)=e^{\eta J_{+}}e^{\mu J_{3}}e^{-\overline{\eta}J_{-}}$. These matrix elements can be computed directly by expanding the exponentials in series and using the actions \eqref{action-1} and \eqref{action-2} on the state vector $\ket{N,m}$. With the identity $\frac{n!}{(n-k)!}=(-1)^{k}(-n)_{k}$, we readily obtain
\begin{equation}
\label{matrix-elements-linear}
\lambda_{k,m}=(-1)^{m}\frac{\rho^{m+k}e^{i\delta(k-m)}}{\sqrt{(1+\rho^2)^{N}}}\binom{N}{k}^{1/2}\binom{N}{m}^{1/2}K_{m}(k;p,N),
\end{equation}
where $p=\frac{\rho^2}{1+\rho^2}$ and $K_{m}(k;p,N)$ is the Krawtchouk polynomial of degree $m$, which has the hypergeometric representation
\begin{equation}
K_m(k;p,N)={}_2F_{1}\left(\begin{aligned} -m, & -k \\ -&N\end{aligned};\frac{1}{p}\right),
\end{equation}
with $m=0,\ldots,N$. This result also follows from the recurrence relation satisfied by the matrix elements $\lambda_{k,m}$. This relation can be obtained from \eqref{recurrence-step-2} or more simply by writing
\begin{equation}
(k-N/2)\,\lambda_{k,m}=\Braket{k,N}{J_{3}D(\eta)}{N,m}=\Braket{k,N}{D(\eta)D^{-1}(\eta)J_{3}D(\eta)}{N,m}.
\end{equation}
Using the B.--C.--H. relation, the matrix elements $\lambda_{k,m}$ are found to obey
\begin{align}
k\lambda_{k,m}=m\left(\frac{1-\rho^2}{1+\rho^2}\right)\lambda_{k,m}+&N\left(\frac{\rho^2}{1+\rho^2}\right)\lambda_{k,m}+\frac{\rho\, e^{i\delta}}{1+\rho^2}\sqrt{(m+1)(N-m)}\;\lambda_{k,m+1}\nonumber\\
&+\frac{\rho\, e^{-i\delta}}{1+\rho^2}\sqrt{m(N-m+1)}\;\lambda_{k,m-1}.
\end{align}
Introducing the polynomials $\lambda_{k,m}=(-\overline{\eta})^{m}\binom{N}{m}^{1/2}K_{m}(k;p,N)\lambda_{k,0}$ with $p=\frac{\rho^2}{1+\rho^2}$, we recover the three-term recurrence relation of the Krawtchouk polynomials
\begin{align}
-k K_{m}(k;p,N)&=-[p(N-m)+m(1-p)]\,K_{m}(k;p,N)\nonumber\\
&+p(N-m)\,K_{m+1}(k;p,N)+m(1-p)\,K_{m-1}(k;p,N).
\end{align}
The coherent state operator $D(\eta)$ is unitary (see \cite{Truax-1985}); consequently, we have the following orthogonality relation between the matrix elements $\lambda_{m,k}$:
\begin{align}
\label{orthogonality-Krawtchouk-1}
\sum_{k=0}^{N}\lambda_{k,m}\lambda^{*}_{k,n}=\delta_{mn},
\end{align}
where $x^{*}$ is the complex conjugate of $x$. In $\lambda_{k,n}$, this amounts to the replacement $\delta\rightarrow-\delta$. In the following, it shall be useful to write the orthogonality relation \eqref{orthogonality-Krawtchouk-1} as the biorthogonality relation
\begin{align}
\label{orthogonality-Krawtchouk-2}
\sum_{k=0}^{N}\lambda_{k,m}\lambda^{*}_{N-n,N-k}=\delta_{nm},
\end{align}
whose equivalence to \eqref{orthogonality-Krawtchouk-1} is shown straightforwardly using the properties of the Krawtchouk polynomials. Expressing the matrix elements as in \eqref{matrix-elements-linear} yields the well-known orthogonality relation of the Krawtchouk polynomials
\begin{align}
\sum_{k=0}^{N}\binom{N}{k}p^k(1-p)^{N-k}K_{n}(k;p,N)K_{m}(k;p,N)=\frac{(-1)^{n}n!}{(-N)_{n}}\left(\frac{1-p}{p}\right)^{n}\delta_{nm}.
\end{align}
\subsection{The matrix elements $\phi_{m,n}$ and vector orthogonal polynomials}
We now turn to the characterization of the matrix elements $\phi_{m,n}$ of the squeezing operator $S(\xi)$. As in the case of the coherent state operator $D(\eta)$, the matrix elements $\phi_{m,n}$ of $S(\xi)=e^{\xi J_{+}^2/2}e^{-\overline{\xi}J_{-}^2/2}$ can be computed directly by expanding the exponentials in series and applying the actions \eqref{action-1} and \eqref{action-2} on the state vector $\ket{N,n}$. Obviously, any matrix element $\phi_{m',n'}$ with $m'$ and $n'$ of different parities will be zero; consequently, we set $n=2a+c$ and $m=2b+c$ with $c=0,1$ and obtain
\begin{align}
\phi_{m,n}=(-1)^{a}\frac{(r/2)^{a+b}e^{i\gamma(b-a)}}{a!b!}\sqrt{\frac{(N-c)!m!}{(N-m)!}}\sqrt{\frac{(N-c)!n!}{(N-n)!}}\,A_{a}^{(c)}(b;d,N),
\end{align}
where the identities $(a)_{2n}=2^{2n}\left(\frac{a}{2}\right)_{n}\left(\frac{a+1}{2}\right)_{n}$ and $(2\sigma+s)!=2^{2\sigma}s!\sigma!(s+1/2)_{\sigma}$ were used and where we have defined
\begin{align}
A_{a}^{(c)}(b;d,N)={}_2F_{3}\left(\begin{matrix}-a & -b & \\
								      c+1/2 & \frac{c-N}{2} &\frac{c-N+1}{2}
                                  \end{matrix};\frac{1}{d}\right),
\end{align}
with $d=-4r^2$. The polynomials $A_{a}^{(c)}$ have been studied in \cite{VXGenest-2011}; we review here some basic results. The polynomials $A_{a}^{(c)}$ are vector orthogonal polynomials of dimension 3, which corresponds to the case $q=1$ and $p=3$ of the general setting presented in the introduction. This can be seen by computing the recurrence relation satisfied by the matrix elements $\phi_{m,n}$. Once again, we start with
\begin{align}
\label{recurrence-vector-polynomials-1}
(m-N/2)\phi_{m,n}=\Braket{m,N}{J_{3}S(\xi)}{N,n}=\Braket{m,N}{S(\xi)S^{-1}(\xi)J_{3}S(\xi)}{N,n}.
\end{align}
Using the B.--C.--H. relation and the formulas from Appendix A, we obtain
\begin{align}
S^{-1}(\xi)J_{3}S(\xi)=(J_{3}+\overline{\xi}J_{-}^2)+\xi[J_{+}-\overline{\xi}(1+2J_{3})J_{-}-\overline{\xi}^{2}J_{-}^3]^2.
\end{align}
Substituting this result into \eqref{recurrence-vector-polynomials-1} yields
\begin{align}
(m-n)\phi_{m,n}&=\xi\sqrt{(n+1)_{2}(n-N)_{2}}\,\phi_{m,n+2}+\overline{\xi}\sqrt{(-n)_{2}(N-n+1)_{2}}\,\phi_{m,n-2}\nonumber\\
&+\xi\overline{\xi}\sum_{j=0}^{3}\overline{\xi}^{j}\sqrt{(-n)_{2j}(N-n+1)_{2j}}f_{n}^{(j)}\,\phi_{m,n-2j},
\end{align}
with coefficients
\begin{align}
f_{n}^{(0)}&=(N-2n)(-1+N+2Nn-2n^2),\\
f_{n}^{(1)}&=(6n^2-12n+N(5-6n)+N^2+9),\\
f_{n}^{(2)}&=(4n-2N-8),\\
f_{n}^{(3)}&=1.
\end{align}
Setting $\phi_{2b+c,2a+c}=\frac{(-\overline{\xi})^{a}}{a!}\sqrt{\frac{(N-c)!n!}{(N-n)!}}A_{a}^{(c)}(b;d,N)\phi_{2b+c,c}$, the polynomials $A_{a}^{(c)}(b;d,N)$ are seen to obey the recurrence relation
\begin{align}
(b-a)A_{a}^{(c)}(b;d,N)&=\frac{-\xi\overline{\xi}}{a+1}(n+1)_{2}(n-N)_{2}A_{a+1}^{(c)}(b;d,N)-aA_{a-1}^{(c)}(b;d,N)\nonumber\\
&+\xi\overline{\xi}\sum_{j=0}^{3}(-a)_{j}f_{n}^{(j)}A_{a-j}^{(c)}(b;d,N).
\end{align}
The matrix elements $\phi_{m,n}$ are thus given by two families of polynomials $A_{a}^{(c)}$ for $c=0,1$ of vector orthogonal polynomials of order $3$. The matrix elements of the inverse operator $S^{-1}(\xi)$ can also be found by direct computation or by inspection. One readily sees that the matrix elements $\phi_{m,n}$ obey the biorthogonality relation
\begin{equation}
\sum_{m=0}^{N}\phi_{m,n}\phi^{*}_{N-n',N-m}=\delta_{nn'},
\end{equation}
In contrast to the situation with the matrix elements $\lambda_{k,m}$ of the displacement operator $D(\eta)$, the biorthogonality relation for the matrix elements $\phi_{m,n}$ is not equivalent to a standard orthogonality relation; this is a consequence of the non-unitarity of $S(\xi)$ and explains the vector-orthogonal nature of the polynomials $A_{a}^{(c)}(b;d,N)$. From the biorthogonality relation of the matrix elements $\phi_{m,n}$ follow two biorthogonality relations for the polynomials $A_{a}^{(c)}(b;d,N)$; we have, for $N=2u+2c$,
\begin{align}
\sum_{b=0}^{u}(-1)^{b}\binom{u}{b}A_{a}^{(c)}(b;d,N)A_{u-a'}^{(c)}(u-b;d,N)=\frac{a!}{(-u)_{a}[(c+1/2)_{u}]^2}\left(\frac{1}{d}\right)^{u}\delta_{aa'}.
\end{align}
For $N=2u+1$, we find a biorthogonality relation interlacing the two families $c=0$ and $c=1$:
\begin{align}
\sum_{b=0}^{u}(-1)^{b}\binom{u}{b}A_{a}^{(1)}(b;d,N)A_{u-a'}^{(0)}(u-b;d,N)=\frac{a!}{(-u)_{a}[(1/2)_{u}(3/2)_{u}]}\left(\frac{1}{d}\right)^{u}\delta_{aa'}.
\end{align}
\subsection{Full matrix elements and squeezed-coherent states}
The results of the two preceding subsections allow to write explicitly the matrix elements $R_{k,n}$; noting that $\phi_{m,n}$ is automatically zero when $m$ and $n$ have different parities, we have, for $n=2a+c$, the following expression for the full matrix elements:
\begin{align}
\label{full-matrix-elements}
R_{k,n}=\Phi\sum_{b=0}^{\lfloor\frac{N-c}{2}\rfloor}\Theta_bK_{2b+c}(k;p,N)A_{a}^{(c)}(b;d,N),
\end{align}
where we have defined
\begin{align}
\Phi&=\frac{1}{\sqrt{(1+\rho^2)^{N}}}\frac{\eta^{k}(-\overline{\xi}/2)^{a}}{a!}\binom{N}{k}^{1/2}\sqrt{\frac{(N-c)!n!}{(N-n)!}},\\
\Theta_{b} &=\frac{(-\overline{\eta})^{2b+c}(\xi/2)^{b}}{b!}\binom{N}{2b+c}^{1/2}\sqrt{\frac{(N-c)!(2b+c)!}{(N-2b-c)!}}.
\end{align}
The generalized squeezed-coherent states are therefore expressed as the linear combination
\begin{equation}
\ket{\eta,\xi}_{n}=\frac{1}{\mathrm{Norm}}\sum_{k=0}^{N}R_{k,n}\ket{N,k},
\end{equation}
where $\mathrm{Norm}$ is a normalization constant. The expression for the amplitudes simplifies  significantly if one considers the standard squeezed-coherent states in which the operator $R(\eta,\xi)$ acts on the vacuum. Indeed, we have
\begin{equation}
\ket{\eta,\xi}=\frac{1}{\mathrm{Norm}}\sum_{k=0}^{N}\sqrt{\frac{1}{(1+\rho^2)^N}}\binom{N}{k}^{1/2}\eta^{k}\left(\sum_{b=0}^{\lfloor\frac{N}{2}\rfloor}\frac{(\overline{\eta}^2\xi/2)^b}{b!}(-N)_{2b}K_{2b}(k;p,N)\right)\ket{k}.
\end{equation}
 If the squeezing parameter $\xi=re^{i\gamma}$ is set to zero, we recover the standard normalized $\mathfrak{su}(2)$ coherent states
\begin{equation}
\ket{\eta}=\sqrt{\frac{1}{(1+\rho^2)^N}}\sum_{k=0}^{N}\binom{N}{k}^{1/2}\eta^{k}\ket{k}.
\end{equation}
In section 8, the properties of the states $\ket{\eta,\xi}$ will be further investigated; in particular, it will be shown that they exhibit spin squeezing when $N$ is even.
\section{Biorthogonality relation}
Given the symmetry of the matrix elements entering the finite convolution yielding $R_{k,n}$, the matrix elements of the inverse operator $S^{-1}(\xi)D^{-1}(\eta)$ are expected to have a similar behavior. Indeed, one finds that
\begin{align}
\sum_{k=0}^{N}R_{k,n}\tilde{R}_{N-k,N-n'}=\delta_{nn'},
\end{align}
where $\sim$ denotes the replacements $\rho\rightarrow -\rho$ and $r\rightarrow -r$. In terms of the vector polynomials $\Psi_{k,n}=(R_{k,3n},R_{k,3n+1},R_{k,3n+2})^{t}$, this biorthogonality relation takes the form
\begin{align}
\sum_{k=0}^{N}\Psi_{k,n}(\tilde{\Psi}_{N-k,N-n'})^{t}=\delta_{nn'}\mathrm{Id}_{3\times 3}.
\end{align}
This equation can be transformed into a biorthogonality relation for the matrix polynomials $Q_{n}(k)$. Indeed, one has
\begin{align}
\sum_{k=0}^{N}Q_{n}(k)\Psi_{k,0}(\tilde{Q}_{N-n'}(N-k)\tilde{\Psi}_{N-k,0})^{t}=\delta_{nn'}\mathrm{Id}_{3\times 3},
\end{align}
which can be written as
\begin{align}
\sum_{k=0}^{N}Q_{n}(k)W(k)(\tilde{Q}_{N-n'}(N-k))^{t}=\delta_{nn'}\mathrm{Id}_{3\times 3},
\end{align}
with the biorthogonality weight matrix 
\begin{align}
W(k)=
\begin{pmatrix}
R_{k,0}\tilde{R}_{N-k,0} &   R_{k,0}\tilde{R}_{N-k,1}   &  R_{k,0}\tilde{R}_{N-k,2}\\
R_{k,1}\tilde{R}_{N-k,0} &   R_{k,1}\tilde{R}_{N-k,1}   &  R_{k,1}\tilde{R}_{N-k,2}\\
R_{k,2}\tilde{R}_{N-k,0} &   R_{k,2}\tilde{R}_{N-k,1}   &  R_{k,2}\tilde{R}_{N-k,2}
\end{pmatrix}.
\end{align}
Each element of the weight matrix can be computed exactly as a finite sum with the help of the equation \eqref{full-matrix-elements}.
\section{Matrix Orthogonality functionals}
From the recurrence relation \eqref{recurrence-step-2} and the results in \cite{Sorokin-1997}, it is known that there exists a set of three $3\times 3$ matrix of functionals with respect to which the matrix polynomials $Q_{n}(k)$ are orthogonal. To construct these functionals, we consider the recurrence relation satisfied by the matrix elements of the inverse operator $R^{-1}(\eta,\xi)$. These matrix elements are defined as
\begin{align}
R^{-1}_{n,k}=\Braket{n,N}{R^{-1}}{N,k},
\end{align}
and their recurrence relation is obtained by noting that
\begin{align}
(k-N/2)R^{-1}_{n,k}=\Braket{n,N}{R^{-1}J_0}{N,k}=\Braket{n,N}{R^{-1}J_0RR^{-1}}{N,k}.
\end{align}
The quantity $R^{-1}J_{3}R$ has been computed previously; recalling that $J_{\pm}^{\dagger}=J_{\mp}$ and that $J_{3}^{\dagger}=J_{3}$, the recurrence relation for the inverse 3-vector $\Psi^{-1}_{n,k}$ is given by
\begin{align}
\label{recurrence-inverse}
k\,\Psi^{-1}_{n,k}=\sum_{j=-1}^{3}e_{n}^{(j)}\Psi_{n+j,k}^{-1},
\end{align}
where the coefficients $e_{n}^{(j)}$ are $3\times 3$ matrices and $e_{n}^{(-1)}$ as well as $e_{n}^{(3)}$ are respectively lower and upper triangular invertible matrices. Since
\begin{equation}
\sum_{k=0}^{N}R_{m,k}^{-1}R_{k,n}=\sum_{k=0}^{N}\Braket{m,N}{R^{-1}}{N,k}\Braket{k,N}{R}{N,n}=\delta_{nm},
\end{equation}
it follows that
\begin{equation}
\label{ortho-multi}
\sum_{k=0}^{N}\Psi_{k,n}(\Psi^{-1}_{m,k})^{t}=\delta_{nm}\mathrm{Id}_{3\times 3}.
\end{equation}
We may now state the following proposition.
\begin{proposition}
The 3-vector $\Psi_{n,k}^{-1}$ can be expressed as
\begin{align}
\label{propo-1}
\Psi_{n,k}^{-1}=\sum_{i=0}^{2}p_i^{(n)}(k)\kappa_{i}\Xi_{i},
\end{align}
where $\Xi_{i}=\Psi^{-1}_{i,k}$. The $\kappa_{i}$ are $3\times 3$ matrices which depend only on $i$ and $p_{i}^{(n)}(k)$ are polynomials in the variable $k$. Let $n=3\nu+\ell$, for $\ell=0,1,2$; if $i\leqslant \ell$, the degree of the polynomial $p_{i}^{(n)}(k)$ is $\nu$, otherwise it is $\nu-1$.
\end{proposition}
\begin{proof}
We set $n=0$ in the recurrence relation \eqref{recurrence-inverse}, which leads to
\begin{equation}
e_{n}^{(3)}\Psi^{-1}_{3,k}=(k\, \mathrm{Id}_{3\times 3}-e_{0}^{(0)})\Psi_{0,k}^{-1}-e^{(1)}\Psi^{-1}_{1,k}-e^{(2)}\Psi^{-1}_{2,k}.
\end{equation}
Since $e_{n}^{(3)}$ is upper triangular and invertible, it follows that
\begin{equation}
\Psi^{-1}_{3,k}=\sum_{i=0}^{2}p_{i}^{(3)}(k)\kappa_{i}\Psi^{-1}_{i,k},
\end{equation}
where the degree of $p_{i}^{(3)}(k)$ is $1$ for $i=0$ and $0$ for all other indices. This establishes \eqref{propo-1} for $n=0$. The proof is then completed by induction.
\end{proof}
From this proposition, it is natural to define the following matrices of functionals.
\begin{definition}
Let $\mathcal{F}_{i}$ for $i=1,2,3$ be the matrix functionals defined by
\begin{align}
\mathcal{F}_{i}[\cdot]=\sum_{k=0}^{N}[\cdot]\Psi_{k,0}\Xi_{i-1}^{t}.
\end{align}
\end{definition}
With this definition, the relation \eqref{ortho-multi} can be written as
\begin{equation}
\mathcal{F}_{i}[k^{\nu}Q_{n}(k)]=0_{3\times 3}\;\;\text{for }\;\;\nu=0,\ldots,\lfloor\frac{n-i}{3}\rfloor,
\end{equation}
for $i=1,2,3$. The multi-orthogonality of the matrix polynomials $Q_{n}(k)$ has thus been made explicit by the direct construction of the orthogonality functionals.
\section{Difference equation}
The matrix polynomials $Q_{n}(k)$ are bi-spectral; not only do they satisfy a recurrence relation, they also obey a difference equation. In a fashion dual to the approach followed to find the recurrence relation, we observe that
\begin{equation}
(n-N/2)R_{k,n}=\Braket{k,N}{R J_{3}}{N,n}=\Braket{k,N}{R J_{3}R^{-1}R}{N,n}.
\end{equation}
Using once again the B.--C.--H relation and formulas from appendix A, we obtain
\begin{equation}
RJ_{3}R^{-1}=(\mathcal{B}_{0}-\xi \mathcal{B}_{+}^2)-\overline{\xi}[\mathcal{B}_{-}+\xi\mathcal{B}_{+}(1+2\mathcal{B}_0)-\xi^2\mathcal{B}_{+}^{3}]^2,
\end{equation}
with
\begin{align}
\mathcal{B}_{0}&=(1-2p)J_3-\rho e^{i\delta}(1-p)J_{+}-\rho e^{-i\delta}(1-p)J_{-},\\
\mathcal{B}_{+}&=(1-p)J_{+}+2\rho e^{-i\delta}(1-p)J_{3}-pe^{-2i\delta}J_{-},\\
\mathcal{B}_{-}&=-p e^{2i\delta}J_{+}+2\rho e^{i\delta}(1-p)J_3+(1-p)J_{-}.
\end{align}
Expanding these expressions leads to a difference equation of the form
\begin{align}
nR_{k,n}=\sum_{j=-6}^{6}m_{k}^{(j)}R_{k+j,n},
\end{align}
which can be turned into a difference equation for the matrix polynomials $Q_{n}(k)$ which is quite involved and not provided here.

It is interesting to observe that this difference equation contains the same number of terms as the recurrence relation, but has a different symmetry. Indeed, the indices run from $-6$ to $6$ in the difference equation. It indicates that the order in which the coherent state operator and the squeezing operator are presented has an impact on the structure of the associated polynomials. If one defines $R'=S(\xi)D(\eta)$ instead of $R=D(\eta)S(\xi)$, one gets $6\times 6$ matrix multi-orthogonal polynomials satisfying a three-term matrix recurrence relation; these polynomials are not, however, orthogonal, because the condition $(\Gamma_{n}^{(-6)})^{*}=\Gamma_{n}^{(6)}$ is not fulfilled.
\section{Generating functions and ladder relations}
The ordinary $\mathfrak{su}(2)$ coherent states $\ket{\eta}$ can be used to obtain a generating function for the matrix elements $R_{k,n}$. We consider the two-variable function defined by
\begin{equation}
G(x,y)=\frac{1}{(1+\rho^2)^{N}}\sum_{k,n}\binom{N}{k}^{1/2}\binom{N}{n}^{1/2}\overline{x}^{k} y^{n}R_{k,n}.
\end{equation}
Clearly, the function $G(x,y)$ can be viewed as the matrix element of $R$ between the coherent states $\ket{x}$ and $\ket{y}$. Introducing a resolution of the identity, we obtain
\begin{equation}
\label{generating}
G(x,y)=\Braket{x}{R}{y}=\sum_{m=0}^{N}\Braket{x}{D(\eta)}{N,m}\Braket{m,N}{S(\xi)}{y}.
\end{equation}
The first part of this convolution can be evaluated directly using the action of the generators $J_{\pm}$ on the states; not surprisingly, we recover, up to a multiplicative factor, the generating function of the Krawtchouk polynomials; the result is
\begin{equation}
\label{premiere}
\Braket{x}{D(\eta)}{N,m}=\frac{\overline{\eta}^{m}}{(1+\rho^2)^{N}}\binom{N}{m}^{1/2}(1+\eta\overline{x})^{N-m}\left(\frac{(1-p)}{p}\eta\overline{x}-1\right)^{m},
\end{equation}
with $p=\frac{\rho^2}{1+\rho^2}$. The second part in the R.H.S of \eqref{generating} is also readily determined. Setting $m=2t+s$, one finds
\begin{equation}
\label{seconde}
\Braket{m,N}{S(\xi)}{y}=\frac{\xi^{t}y^{s}}{t!}\sqrt{\frac{N!m!}{(N-m)!}}\sum_{k=0}^{\lfloor\frac{N-s}{2}\rfloor}\frac{(-y^2/\xi)^{k}}{(2k+s)!}(-t)_{k}\;{}_2F_{0}\left(\frac{-z(k)}{2},\frac{1-z(k)}{2};-4y^2\overline{\xi}\right),
\end{equation}
where $z(k)=N-2k-s$. The convolution of the two functions given in \eqref{premiere} and \eqref{seconde} thus yields the generating function for the matrix elements $R_{k,n}$. The ladder relations for the matrix polynomials $Q_{n}(k)$ can also be constructed explicitly from the observations
\begin{align}
\sqrt{-(n+1)_{3}(n-N)_{3}}\,R_{k,n+3}&=\Braket{k,N}{RJ_{+}^{3}}{N,n},\nonumber\\
\sqrt{-(-n)_{3}(N-n+1)_{3}}\,R_{k,n-3}&=\Braket{k,N}{RJ_{-}^{3}}{N,n}.
\end{align}
The conjugation of the operator $J_{\pm}^{3}$ by the operator $R$ leads to complicated expressions which are best evaluated with the assistance of a computer.
\section{Observables in the squeezed-coherent states}
We now  further investigate the properties of the states $\ket{\eta,\xi}$ resulting from the application of the squeezed-coherent operator $D(\eta)S(\xi)$ on the vacuum $\ket{N,0}$. For systems which possess the $\mathfrak{su}(2)$ symmetry, there exist many different parameters to determine whether a state is squeezed or not (for a review of the parameters that can be used see \cite{Wang-2011}). In the following, we will adopt \cite{Wang-2001}
\begin{align}
\label{critere}
Z^2_{\vec{n}_{i}}=N\frac{(\Delta J_{\vec{n}_{i}})^2}{\langle J_{\vec{n}_{i+1}}\rangle^2+\langle J_{\vec{n}_{i+2}}\rangle^2},
\end{align}
where the indices are to be understood cyclically. If $Z^2_{\vec{n_i}}<1$, the system is squeezed in the direction $\vec{n}_{i}$, where the $\vec{n}_{i}$, $i=1,2,3$ are orthogonal unit vectors. This choice of the squeezing criterion is relevant because of its relation with entanglement \cite{Wang-2011}.

The mean value of any observable $\mathcal{O}$ in the state $\ket{\eta,\xi}$ can be expressed by
\begin{align}
\langle\mathcal{O}\rangle=\frac{1}{\kappa(r)}\sum_{i,j}\frac{(\xi/2)^{i}(\overline{\xi}/2)^j}{i!j!}\sqrt{(2i)!(2j)!(-N)_{2i}(-N)_{2j}}\;\Braket{2j}{D^{-1}(\eta)\mathcal{O}D(\eta)}{2i},
\end{align}
where $\kappa$ is the normalization constant, which is given by the hypergeometric function
\begin{align}
\kappa(r)=|\braket{\xi,\eta}{\eta,\xi}|^2={}_3F_{0}\left(1/2,\frac{-N}{2},\frac{1-N}{2}; (2r)^2\right).
\end{align}
For simplicity, we choose to study squeezing along the axis of the Hamiltonian $J_{3}$. We begin by setting $r=0$ to study the behavior of the parameter $Z_{3}^{2}$ in the coherent states. We readily find
\begin{align}
\langle J_1\rangle&=\langle Q\rangle=\frac{N\rho}{1+\rho^2}\cos(\delta) &&
\langle J_2\rangle=-\langle P\rangle=-\frac{N\rho}{1+\rho^2}\sin(\delta),\\
\langle J_3\rangle&=\langle (H-1/2)\rangle=-\frac{N}{2}\frac{1-\rho^2}{1+\rho^2},&&
(\Delta J_{3})^2=\langle J_{3}^2\rangle-\langle J_{3}\rangle^2=\frac{N\rho^2}{(1+\rho^2)^2}.
\end{align}
From these results, it is seen that for pure coherent states $\ket{\eta}$, we always have $Z^2_{3}=1$, which ensures, according to the definition \eqref{critere}, that purely coherent states are never squeezed; this can be proved in a straightforward manner for any choice of normalized basis $\{\vec{n}_1,\vec{n}_2,\vec{n}_3\}$\cite{Wang-2001}.

We now investigate squeezing in the $\vec{n}_{3}=(0,0,1)$ direction for the states $\ket{\eta,\xi}$. We have that
\begin{align}
\langle J_{1}\rangle&=\left[\frac{\rho\cos(\delta)}{1+\rho^2}\right]G_{n}(r),
&&\langle J_{2}\rangle=\left[\frac{-\rho\sin(\delta)}{1+\rho^2}\right]G_{n}(r),
&&&\langle J_{3}\rangle=\left[\frac{1-\rho^2}{1+\rho^2}\right]H_{n}(r)
\end{align}
\begin{align}
\langle J_{3}^2\rangle&=\frac{1}{(1+\rho^2)^2}\left(\left[1+\rho^4\right]J_{n}(r)+\rho^2L_{n}(r)+2\rho^2[\cos(2\delta-\gamma)]M_{n}(r)\right)
\end{align}
where we have defined the polynomials
\begin{align}
G_{n}(r)&=\frac{1}{\kappa(r)}\sum_{i=0}^{\lfloor \frac{N}{2}\rfloor}\frac{(r^2)^{i}}{i!}(1/2)_{i}(-N)_{2i}(N-4i),\\
H_{n}(r)&=\frac{1}{\kappa(r)}\sum_{i=0}^{\lfloor \frac{N}{2}\rfloor}\frac{(r^2)^{i}}{i!}(1/2)_{i}(-N)_{2i}(2i-N/2),
\end{align}
\begin{align}
J_{n}(r)&=\frac{1}{\kappa(r)}\sum_{i=0}^{\lfloor \frac{N}{2}\rfloor}\frac{(r^2)^i}{i!}(1/2)_{i}(-N)_{2i}(4i^2-2iN+N^2/4),\\
L_{n}(r)&=\frac{1}{\kappa(r)}\sum_{i=0}^{\lfloor \frac{N}{2}\rfloor}\frac{(r^2)^i}{i!}(1/2)_{i}(-N)_{2i}(-16i^2+8iN-N^2/2+N)\\
M_{n}(r)&=\frac{r}{\kappa(r)}\sum_{i=0}^{\lfloor \frac{N}{2}\rfloor}\frac{(r^2)^i}{i!}(1/2)_{i+1}(-N)_{2i+2}
\end{align}
The exact expression for the squeezing parameter $Z^2_{\vec{n}_{z}}$ cannot be obtained in closed form; nevertheless, it can easily be computed numerically. We observed that squeezing ($Z^2_{\vec{n}_{z}}<1$) along the Hamiltonian axis $\vec{n}_{z}$ occurs only in the case where $N$ is even, which corresponds to an oscillator having an odd number of points.  It is worth noting that a distinction between the $N$ odd and  the $N$ even cases also arises in the Fourier-Krawtchouk transform \cite{Wolf-2001}, which transforms the finite oscillator wave functions into themselves. The squeezing parameter is plotted against $\theta$ in figure 1.
\begin{figure}[!ht]
\begin{center}
{\includegraphics[width=0.4\linewidth]{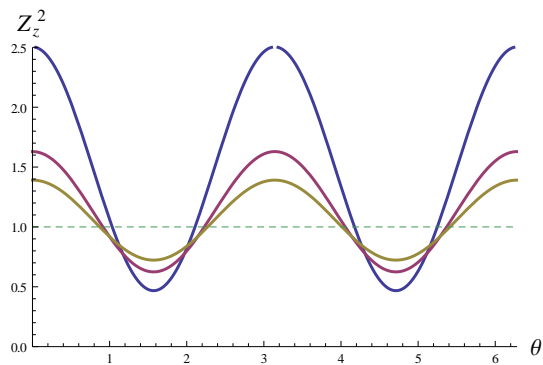}}
\end{center}
\caption{Squeezing parameter $Z^2_{\vec{n}_{z}}$ for $r=2,4,6$ (decreasing amplitudes) with $\rho=0.8$ and $N=40$.}
\end{figure}

\newpage
 \section{Contraction to the standard oscillator}
It is of interest to study the behavior of the polynomials $Q_{n}(k)$ and of the squeezed-states in the contraction limit $N\rightarrow \infty$ where the finite $\mathfrak{u}(2)$ oscillator tends to the standard quantum harmonic oscillator; this limit was studied in detail in \cite{Wolf-2003}. To obtain the proper limit, the parameters of the squeeze-coherent operator  must be renormalized according to
\begin{align}
\rho\rightarrow \frac{\rho}{\sqrt{N}},\;\;\;\; r\rightarrow \frac{r}{N}.
\end{align}
Upon taking the limit, one finds that the 13-term recurrence relation of the $R_{n,k}$ tends to a 5-term symmetric recurrence relation of the form
\begin{equation}
kR_{k,n}=\sum_{j=-2}^{2}c_{n}^{(j)}R_{k,n+j}.
\end{equation}
This is indeed the type of recurrence relation obtained in \cite{Vinet-2011}. In the case of the standard oscillator, it is possible to choose a unitary version of the corresponding $S$ operator and to express its matrix elements in terms of $2\times 2$ matrix orthogonal polynomials.
\section{Conclusion}
The matrix elements of the exponential operators corresponding to the squeezed-coherent operators of the finite oscillator have been determined in the energy eigenbasis of this model. They were seen to be given in terms of matrix multi-orthogonal polynomials which have the Krawtchouk and vector orthogonal polynomials as building blocks. The algebraic setting allowed to characterize these polynomials and to explicitly compute their matrix orthogonality functionals. The results have been used to show that the squeezed coherent states of the finite oscillator exhibit squeezing when the dimension of the oscillator $N+1$ is odd. 
\section*{Acknowledgments}
The authors are very grateful to K.B. Wolf for helpful comments. The authors also wish to thank the referees for drawing our attention on \cite{Spindel-2008} and for their suggestions. The research of L.V. is supported in part through a grant from the National Sciences and Engineering Council of Canada (NSERC). A.Z. wishes to thanks the Centre de recherches math\'ematiques (CRM) for its hospitality during the course of this investigation.
\appendix
\section*{Appendix A--Useful formulas involving the $\mathfrak{su}(2)$ generators}
The relation
\begin{equation*}
J_{3}J_{\pm}^n=J_{\pm}^{n}(J_3\pm n),
\end{equation*}
holds and  can be proven straightforwardly by induction on $n$. Using this identity and the relations $(J_{\pm})^{\dagger}=J_{\mp}$ as well as $J_3^{\dagger}=J_{3}$ , it follows that for $P(J_{\pm})$ denoting a polynomial in $J_{\pm}$, one has
$$
[P(J_{\pm}),J_3]=\mp J_{\pm}P'(J_{\pm}),
$$
where $P'(x)$ is the derivative of $P(x)$ with respect to $x$. The preceding formula and the Baker-Campbell-Hausdorff relation lead to the identity
$$
e^{P(J_{\pm})}J_3e^{-P(J_{\pm})}=J_3\mp J_{\pm}P'(J_{\pm}).
$$
In addition, we have the relations
\begin{align*}
[J_{+},J_{-}^{n}]&=2nJ_{3}J_{-}^{n-1}+n(n-1)J_{-}^{n-1},\\
[J_{-},J_{+}^{n}]&=-2nJ_{+}^{n-1}J_{3}-n(n-1)J_{+}^{n-1},
\end{align*}
which can also be proved by induction on $n$. With the help of the previous identities, one obtains
\begin{align*}
[J_{+},P(J_{-})]&=2J_{3}P'(J_{-})+J_{-}P''(J_{-}),\\
[J_{-},P(J_{+})]&=-2P'(J_{+})J_{3}-J_{+}P''(J_{+}),
\end{align*}
From these formulas it follows that
\begin{align*}
e^{P(J_{-})}J_+e^{-P(J_{-})}=J_+ -2J_{3}P'(J_{-})-J_{-}[P''(J_{-})+P'(J_{-})^2],\\
e^{P(J_{+})}J_-e^{-P(J_{+})}=J_- +2P'(J_{+})J_{3}+J_{+}[P''(J_{+})-P'(J_{+})^2].
\end{align*}
\section{Krawtchouk polynomials}
The Krawtchouk polynomials have the hypergeometric representation
$$
K_{n}(x;p,N)={}_2F_{1}\left[\begin{aligned}-n&,-x\\ -&N\end{aligned};\frac{1}{p}\right].
$$
Their orthogonality relation is
$$
\sum_{x=0}^{N}\binom{N}{x}p^{x}(1-p)^{N-x}K_m(x;p,N)K_{n}(x;p,N)=\frac{(-1)^{n}n!}{(-N)_n}\left(\frac{1-p}{p}\right)^{n}\delta_{nm}.
$$
They have the generating function
$$
(1+t)^{N-x}\left(1-\frac{1-p}{p}t\right)^{x}=\sum_{n=0}^{N}\binom{N}{n}K_{n}(x;p,N)t^{n}.
$$
For further details, see \cite{Koekoek-2010}.


\begin{thebibliography}{10}

\bibitem{Wolf-2001}
N.M. Atakishiyev, G.S. Pogosyan, L.E. Vicent, and K.B. Wolf.
\newblock Finite two-dimensional oscillator {I}: The cartesian model.
\newblock {\em Journal of Physics A: Mathematical and General}, 34:9381--9398,
  2001.

\bibitem{Wolf-2003}
N.M. Atakishiyev, G.S. Pogosyan, and K.B. Wolf.
\newblock Contraction of the finite one-dimensional oscillator.
\newblock {\em International Journal of Modern Physics A}, 18:317--327, 2003.

\bibitem{VXGenest-2011}
V.X. Genest, L.~Vinet, and A.~Zhedanov.
\newblock \textit{d}-{O}rthogonal polynomials and $\mathfrak{su}(2)$.
\newblock {\em Journal of Mathematical Analysis and Applications},
  390:472--487, 2012.

\bibitem{Koekoek-2010}
R.~Koekoek, P.A. Lesky, and R.F. Swarttouw.
\newblock {\em Hypergeometric orthogonal polynomials and their q-analogues}.
\newblock Springer, 1\textsuperscript{st} edition, 2010.

\bibitem{Beckermann-1992}
G.~Labahn and B.~Beckermann.
\newblock A uniform approach for {H}ermite-{P}ad{\'e} and simultaneous
  {P}ad{\'e} approximants and their matrix type generalization.
\newblock {\em Numerical Algorithms}, 3:45--54, 1992.

\bibitem{Wang-2011}
J.~Ma, X.~Wang, C.P. Sun, and F.~Nori.
\newblock Quantum spin squeezing.
\newblock {\em Physics Reports}, (509):89--165, 2011.

\bibitem{Spindel-2008}
S.~Massar and P.~Spindel.
\newblock Uncertainty relation for the discrete fourier transform.
\newblock {\em Phys. Rev. Lett.}, 100:190401, 2008.

\bibitem{Niederer-1972}
U.~Niederer.
\newblock The maximal kinematical invariance group of the free
  {S}chr{\"o}dinger equations.
\newblock {\em Helvetica Physica Acta}, 45:802--810, 1972.

\bibitem{Niederer-1973}
U.~Niederer.
\newblock Maximal kinematical invariance group of the harmonic oscillator.
\newblock {\em Helvetica Physica Acta}, 46:191--200, 1973.

\bibitem{Satyanarayana-1985}
M.V. Satyanarayana.
\newblock Generalized coherent states and generalized squeezed coherent states.
\newblock {\em Physical Review D}, 32:400--404, 1985.

\bibitem{Sorokin-1997}
V.N. Sorokin and J.~Van Iseghem.
\newblock Algebraic aspects of matrix orthogonality for vector polynomials.
\newblock {\em Journal of Approximation Theory}, 90:97--116, 1997.

\bibitem{Truax-1985}
D.R. Truax.
\newblock {B}aker-{C}ampbell-{H}ausdorff relations and unitarity of su(2) and
  su(1,1) squeeze operators.
\newblock {\em Physical Review D}, 31:1988--1991, 1985.

\bibitem{Vinet-2011}
L.~Vinet and A.~Zhedanov.
\newblock Representations of the {S}chr{\"o}dinger group and matrix orthogonal
  polynomials.
\newblock {\em Journal of Physics A: Mathematical and Theoretical},
  44(35):1--28, 2011.

\bibitem{Wang-2001}
X.~Wang.
\newblock Spin squeezing in nonlinear spin-coherent states.
\newblock {\em Journal of Optics B: Quantum and Semiclassical Optics},
  3:93--96, 2001.

\bibitem{Wigner-1956}
E.P. Wigner and E.~In{\"o}n{\"u}.
\newblock On the contraction of groups and their representations.
\newblock {\em Proceedings of the National Academic Society}, 39:510--524,
  1956.

\bibitem{Wolf-2007}
K.B. Wolf and G.~Kr{\"o}tzsch.
\newblock Geometry and dynamics of squeezing in finite systems.
\newblock {\em Journal of the Optical Society of America}, 24:2871--2878, 2007.

\end{thebibliography}
\end{document}